\documentclass[12pt]{article}
\usepackage{amsmath,amssymb,amsfonts,enumerate,amsthm}
\usepackage[dvipdfm,dvipdf,dvips,pdftex]{graphics,graphicx,color}
\usepackage[hmargin=3cm,vmargin=3cm]{geometry}
\newtheorem{theorem}{Theorem}[section]

\newtheorem{definition}[theorem]{Definition}
\newtheorem{corollary}[theorem]{Corollary}
\newtheorem{remark}[theorem]{Remark}
\newtheorem{example}[theorem]{Example}
\begin{document}
\title{ Some results of domination and total domination in the direct product of two fuzzy graphs }
\author{
Pradip Debnath\thanks{Telephone No. +91 9085 244601, Fax: +91 3842 233797}\\
Department of Mathematics\\
National Institute of Technology Silchar\\
Silchar 788010, Assam\\
India\\
\texttt{debnath.pradip@yahoo.com}}
\date{}
\maketitle
\begin{abstract}
In this article we give a new definition of direct product of two arbitrary fuzzy graphs. We define the concepts of domination and total domination in this new product graph. We obtain an upper bound for the total domination number of the product fuzzy graph. Further we define the concept of total $\alpha$-domination number and derive a lower bound for the total domination number of the product fuzzy graph in terms of the total $\alpha$-domination number of the component graphs. A lower bound for the domination number of the same has also been found.
\end{abstract}
\textbf{Keywords: } Direct product; Fuzzy graph; Domination number; Total domination number.  \\
\textbf{Mathematics Subject Classification: } 05C72; 03E72; 03F55    \\
\section{Introduction with preliminaries}
Mathematical research on the theory of domination for crisp graph was initiated by Ore~\cite{ore62} in 1962 which was further developed by Cockayne and Hedetnieme~\cite{cock77}. The notion of fuzzy graph was introduced by Rosenfeld~\cite{ros75} in 1975. Somasundaram and Somasundaram~\cite{som98} introduced and studied the theory of domination in fuzzy graphs.

In graph theory, four most important types of standard products are Cartesian product, direct product, strong product and lexicographic product. Among all other associative products, only these four products preserve the properties of the factor or component graphs in a meaningful way. Also, these four products have the property that at least one projection is weak homomorphism. The theory of product graph has wide application in science and engineering, e.g., such application may be found in~\cite{lamp74} where this theory was used to model concurrency in multiprocessor systems. This can also be used in automata theory. For more information on such products and applications we refer to~\cite{hamm11, lamp74}. The Cartesian product of two arbitrary fuzzy graphs was defined by Mordeson and Cheng~\cite{mord94}. The theory of domination has not been studied previously in the setting of direct product of fuzzy graphs. Motivated by this fact, in this paper we give a new definition of the direct product of fuzzy graphs and study the theory of domination and total domination in the resulting product fuzzy graph. 

Generally, it is quite difficult to determine the exact domination or total domination number of a product graph. Thus obtaining upper or lower bounds of the same is also of great interest. A very useful literature and survey on the theory of domination of crisp graphs have been detailed by Haynes, Hedetniemi and Slater~\cite{hayn98,hay98}. Some other important work on the theory of domination and total domination may be found in~\cite{nath13, dorb06, moh10, mord98}. For crisp graph theoretic terminologies, the readers are referred to Hararay~\cite{har72}.

Throughout this article, by a fuzzy graph we mean an undirected simple fuzzy graph (i.e., graph without loop or multiple edge) which does not have any isolated vertex, unless otherwise stated.

A fuzzy graph $G=(V(G),E(G),\sigma,\mu)$ consists of two functions $\sigma:V(G)\rightarrow [0,1]$ and $\mu:E(G)\rightarrow [0,1]$ such that $\mu (\{x,y\})\leq \min\{\sigma(x),\sigma (y)\}$ for all $x,y\in V(G)$.

The set $N(x)=\{y\in V(G):\mu (\{x,y\})=\min\{\sigma(x),\sigma (y)\}$ is called the open neighborhood of $x$ and $N[x]=N(x)\cup \{x\}$ is called the closed neighborhood of $x$.

A fuzzy graph $G=(V(G),E(G),\sigma,\mu)$ is said to be complete if $$\mu (\{x,y\})=\min\{\sigma(x),\sigma (y)\}$$ for all $x,y\in V(G)$.

We say that the vertex $x$ dominates $y$ in $G$ if $\mu (\{x,y\})=\min\{\sigma(x),\sigma (y)\}$. A subset $S$ of $V(G)$ is called a dominating set in $G$ if for every $v\notin S$, there exists $u\in S$ such that $u$ dominates $v$. The fuzzy cardinality of $S$ is defined as $\sum_{v\in S}\sigma (v)$. The minimum fuzzy cardinality of a dominating set in $G$ is called the domination number of $G$ and denoted by $\nu(G)$.

A subset $T$ of $V(G)$ is said to be a total dominating set if every vertex in $V(G)$ is dominated by a vertex in $T$. The minimum fuzzy cardinality of a total dominating set is called the total domination number and denoted by $\nu_{t}(G)$. Such a dominating set with minimum fuzzy cardinality is called a minimal dominating set of $G$.

If $g$ is a vertex of $G$, then by the symbol $~^{g}H$ we denote the subgraph of $G\times H$ induced by $\{g\}\times V(H)$. The subgraph $G^{h}$ may be defined in a similar fashion.

Let $V$ be a nonempty set. The weight of a function $f:V\rightarrow \mathbb{R}$ is defined by $w(f)=\sum_{v\in V}f(v)$ and for a nonempty subset $S\subseteq V$, we define $f(S)=\sum_{v\in S}f(v)$.

\section{Domination and total domination in direct product of fuzzy graphs}

Now we discuss our main results.

\begin{definition}
Let $G=(V(G),E(G),\sigma_{1},\mu_{1})$ and $H=(V(H),E(H),\sigma_{2},\mu_{2})$ be two fuzzy graphs. The direct product of $G$ and $H$ is denoted by $G\times H$ and the vertex set of $G\times H$ is defined by $V(G\times H)=V(G)\times V(H)$ while edge set of $G\times H$ is denoted and defined by  
$$E(G\times H)=(V(G)\times V(H))\times (V(G)\times V(H))
=\{\{g_{1}h_{1},g_{2}h_{2}\}:\{g_{1},g_{2}\}\in E(G),\{h_{1},h_{2}\}\in E(H)\}.$$
 
We use the notation $g_{1}g_{2}$ to denote the edge $\{g_{1},g_{2}\}$ in $E(G)$ and use the notation $gh$ to denote the vertex $(g,h)\in V(G\times H)$ when no confusion arises.

Define $\Lambda:V(G)\times V(H)\rightarrow [0,1]$ by $$\Lambda(g,h)=\min \{\sigma_{1}(g),\sigma_{2}(h)\},$$
and $\Gamma:E(G\times H) \rightarrow [0,1]$ by $$\Gamma(g_{1}h_{1},g_{2}h_{2})=\min\{\mu_{1}(g_{1}g_{2}),\mu_{2}(h_{1}h_{2})\},$$

Then we say that $G\times H=(V(G\times H),E(G\times H),\Lambda,\Gamma)$ is the direct product fuzzy graph if $$\Gamma(g_{1}h_{1},g_{2}h_{2})\leq \min\{\Lambda(g_{1}h_{1}),\Lambda(g_{2}h_{2})\}$$ for all $\{g_{1}h_{1},g_{2}h_{2}\}\in E(G\times H)$. Now on we use the term product fuzzy graph to denote the direct product fuzzy graph.

We say that the vertex $g_{1}h_{1}$ dominates the vertex $g_{2}h_{2}$ in $G\times H$ if $$\Gamma(g_{1}h_{1},g_{2}h_{2})= \min\{\Lambda(g_{1}h_{1}),\Lambda(g_{2}h_{2})\}$$.

Order of the product graph $G\times H$ is denoted and defined by $$p=\sum_{(g,h)\in V(G\times H)}\Lambda(g,h).$$
\end{definition}

The following result is a direct consequence of above definition.

\begin{theorem}
If $g_{1}$ dominates $g_{2}$ in $G$ and $h_{1}$ dominates $h_{2}$ in $H$, then the vertex $g_{1}h_{1}$ dominates the vertex $g_{2}h_{2}$ in $G\times H$. 
\end{theorem}

\begin{definition}
Let $G\times H=(V(G\times H),E(G\times H),\Lambda,\Gamma)$ be a product fuzzy graph. A subset $D\subseteq V(G\times H)$ is called a dominating set if for all $(g,h)\in V(G\times H)\setminus D$, there exists $(g^{'},h^{'})\in D$ such that $(g^{'},h^{'})$ dominates $(g,h)$. The fuzzy cardinality of $D$ is denoted and defined by $fc(D)=\sum_{(g,h)\in D}\Lambda(g,h)$.

The minimum fuzzy cardinality taken over all dominating sets in $G\times H$ is called the dominating number of $G\times H$ and denoted by $\nu(G\times H)$.

A subset $D\subseteq V(G\times H)$ is called a total dominating set if every vertex $(g,h)\in V(G\times H)$ is dominated by a vertex $(g^{'},h^{'})\in D$. 

The minimum fuzzy cardinality taken over all total dominating sets in $G\times H$ is called the total dominating number of $G\times H$ and denoted by $\nu_{t}(G\times H)$.
\end{definition}

It is obvious that for the existence of a total dominating set in $G\times H$ it is necessary that every vertex of $G\times H$ be dominated by some other vertex. Thus we have the following:

\begin{theorem}
The product fuzzy graph $G\times H$ has a total dominating set if and only if $G$ and $H$ both have a total dominating set.
\end{theorem}

In the following theorem we find an upper bound for the total domination number of the product fuzzy graph.

\begin{theorem}
For any two fuzzy graphs $G$ and $H$, $$\nu_{t}(G\times H)\leq \min\{|D_{2}|\nu_{t}(G),|D_{1}|\nu_{t}(H)\}$$
where $D_{1}$ and $D_{2}$ are minimal total dominating sets for $G$ and $H$ respectively and $|A|$ denotes the crisp cardinality of a set $A$.
\end{theorem}
\begin{proof}
Suppose $D=D_{1}\times D_{2}\subseteq V(G\times H)$. Let $(g,h)$ be a vertex of $G\times H$. Then there exist $g^{'}\in D_{1}$ and $h^{'}\in D_{2}$ such that $g^{'}$ dominates $g$ in $G$ and $h^{'}$ dominates $h$ in $H$. 

Clearly, in this case $\Gamma(gh,g^{'}h^{'})= \min\{\Lambda(gh),\Lambda(g^{'}h^{'})\}$. Thus $(g^{'},h^{'})$ dominates $(g,h)$, i.e., $(g^{'},h^{'})\in D$. Therefore $D$ is a total dominating set for $G\times H$ and as such we have $$\nu_{t}(G\times H)\leq \sum_{(g,h)\in D}\Lambda(g,h)\leq |D_{2}|\nu_{t}(G).$$

In a similar fashion, it can be proved that $\nu_{t}(G\times H)\leq \sum_{(g,h)\in D}\Lambda(g,h)\leq |D_{1}|\nu_{t}(H).$
\end{proof}

Proof of the following theorem is trivial and hence omitted.

\begin{theorem}
For any product fuzzy graph $G\times H$, $\nu_{t}(G\times H)=p$ if and only if every vertex of $G$ has a unique neighbor.
\end{theorem}
\begin{corollary}
If $\nu_{t}(G\times H)=p$, then the number of vertices in $G\times H$ must be even.
\end{corollary}

\begin{definition}
A product fuzzy graph $G\times H$ is called complete if for any $(g_{1},h_{1}),(g_{2},h_{2})\in V(G\times H)$ such that $g_{1}\neq g_{2}$ and $h_{1}\neq h_{2}$, we have $\Gamma(g_{1}h_{1},g_{2}h_{2})= \min\{\Lambda(g_{1}h_{1}),\Lambda(g_{2}h_{2})\}$. We denote it by $K^{G\times H}_{\Lambda}$.
\end{definition}

\begin{example}
Below we give example of a complete product fuzzy graph.
	

\begin{center}
\includegraphics[scale=.6]{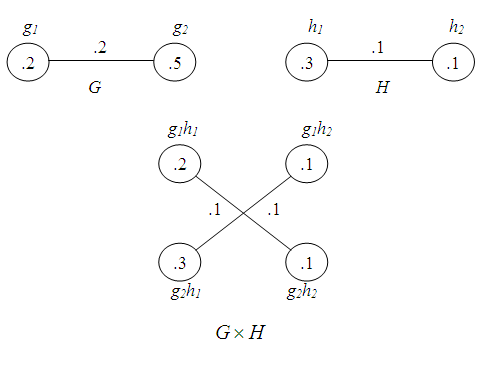}
\end{center}

Here both $G$ and $H$ are complete fuzzy graphs and hence their product $G\times H$ is complete as well. Also, each vertex of $G\times H$ has a unique neighbor and clearly $V(G\times H)$ is the only total dominating set. Hence $\nu_{t}(G\times H)=p=0.7$. Here $\{g_{1}h_{1},g_{1}h_{2}\}$ is the only minimal dominating set of $G\times H$ and therefore $\nu(G\times H)=0.3$.
\end{example}

\begin{theorem}
If both $G$ and $H$ are complete fuzzy graphs, then the product fuzzy graph $G\times H$ is complete as well.
\end{theorem}

Next example suggests that converse of previous theorem is not true in general.

\begin{example}
Consider the fuzzy graphs $G$ and $H$ and their product fuzzy graph $G\times H$.
\begin{center}
\includegraphics[scale=.5]{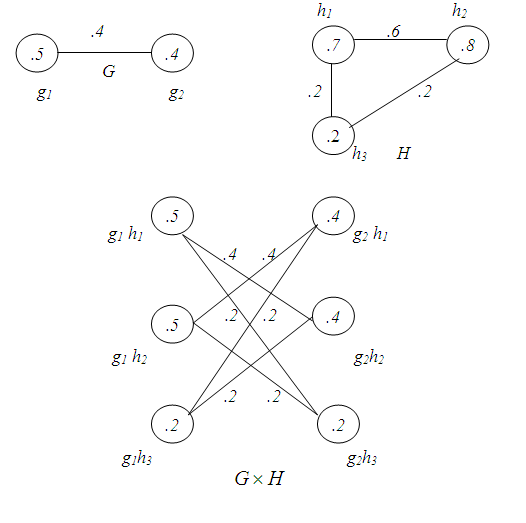}
\end{center}

Here the product graph $G\times H$ is complete while $G$ is complete but $H$ is not complete.
\end{example}

Next example shows that every product fuzzy graph may not have a total dominating set.

\begin{example}
Consider the following fuzzy graphs $G$ and $H$ and their product fuzzy graph $G\times H$.

\begin{center}
\includegraphics[scale=.5]{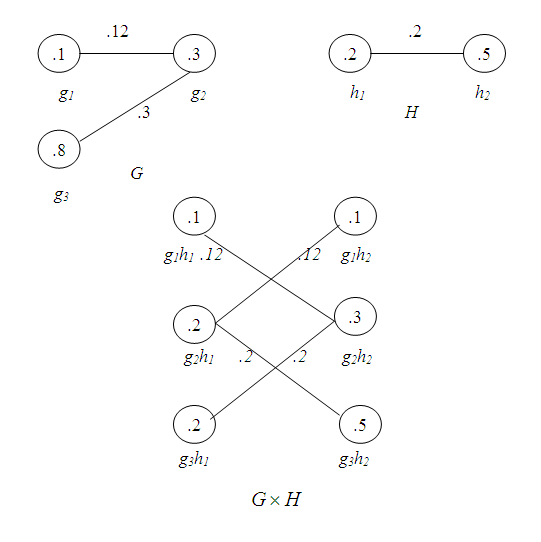}
\end{center}

Here $\{g_{1}h_{1},g_{1}h_{2},g_{2}h_{1},g_{2}h_{2}\}$ is the minimal dominating set of $G\times H$ and $\nu(G\times H)=0.7$. Also, here $g_{2}h_{1}$ dominates $g_{3}h_{2}$ and $g_{2}h_{2}$ dominates $g_{3}h_{1}$. Clearly, $G\times H$ does not have a total dominating set.
\end{example}

\begin{remark}
For a complete product fuzzy graph $K^{G\times H}_{\Lambda}$, for each $g_{1}\in V(G)$, the collection $\{\{g_{1},h\}:h\in H\}$ is a dominating set and also for each $h_{1}\in V(H)$, the collection $\{\{g,h_{1}\}:g\in G\}$ is a dominating set.
\end{remark}

\begin{theorem}
 $\nu(G\times H)=p$ if and only if $V(G\times H)$ is the only dominating set, i.e., if $\Gamma(g_{1}h_{1},g_{2}h_{2})<\min\{\Lambda(g_{1}h_{1}),\Lambda(g_{2}h_{2})\}$ for all $(g_{1},h_{1}),(g_{2},h_{2})\in V(G\times H)$.
\end{theorem}

Now we define the notion of total $\alpha$-domination number of a fuzzy graph and find its relation with total domination number of product fuzzy graph.

\begin{definition}
Let $G$ be a fuzzy graph. For $\alpha>0$, a function $f:V(G)\rightarrow \mathbb{R}$ is called a total $\alpha$-dominating function if for every $v\in V$, $f(N(v))\geq \alpha$.

The total $\alpha$-domination number $\gamma^{\alpha}_{t}(G)$ of a fuzzy graph $G$ is the minimum weight of a total $\alpha$-dominating function.
\end{definition}

The following theorem gives a lower bound for the total domination number of the product fuzzy graph in terms of the total $2\alpha$-domination number of the component fuzzy graphs.

\begin{theorem}
For any two nontrivial connected fuzzy graphs $G=(V(G),\sigma_{1},\mu_{1})$ and $H=(V(H),\sigma_{2},\mu_{2})$ with $\sigma_{1}(g)\geq \alpha$ for all $g\in V(G)$ and $\sigma_{2}(h)\geq \alpha$ ($\alpha>0$) for all $h\in V(H)$, we have $$\nu_{t}(G\times H)\geq \max \{\gamma^{2\alpha}_{t}(G),\gamma^{2\alpha}_{t}(H)\}.$$
\end{theorem}
\begin{proof}
Let $S$ be a minimal total dominating set of $G\times H$. Define $f:V(G)\rightarrow \mathbb{R}$ such that $$f(g)=\min\{2\alpha, fc(S\cap ~^{g}H)\}$$ (Here $fc(S\cap ~^{u}H)$ denotes the fuzzy cardinality of the set $S\cap ~^{u}H$)

We prove that $f$ is a total $2\alpha$ dominating function. Let $g\in V(G)$ and $V(H)=\{h_{1},h_{2},\ldots,h_{n}\}$. Since $S$ is a total dominating set, there exists a vertex $(x,h_{i})\in S$ which dominates $(g,h_{1})$. Clearly, $x\neq g$ and $i\neq 1$ for $G$ and $H$ both are simple graphs. Now consider the vertex $(g,h_{i})$ which is dominated by some $(y,h_{j})\in S$. Here $y\neq g$ and $i\neq j$.

Now if $x=y$, since $i\neq j$, we have $(x,h_{i}),(x,h_{j})\in S\cap ~^{x}H$. Thus $f(x)\geq 2\alpha$ and since $x$ is a neighbor of $g$, we have $f(N(g))\geq 2\alpha$.

Again if $x\neq y$, since $i\neq j$, we have $f(x)\geq \alpha$ (for $g$ is a neighbor of $x$) and $f(y)\geq \alpha$ (for $g$ is a neighbor of $y$). Thus in this case also, $f(N(g))\geq 2\alpha$.

Therefore $f$ is a total $2\alpha$-dominating function with $w(f)\leq fc(S)$  and hence $\nu_{t}(G\times H)\geq \gamma^{2\alpha}_{t}(G)$.

Similarly it can be proved that $\nu_{t}(G\times H)\geq \gamma^{2\alpha}_{t}(H)$.
\end{proof}

\begin{theorem}
For any two nontrivial connected fuzzy graphs $G=(V(G),\sigma_{1},\mu_{1})$ and $H=(V(H),\sigma_{2},\mu_{2})$ with $\sigma_{1}(g)\geq \alpha$ for all $g\in V(G)$ and $\sigma_{2}(h)\geq \alpha$ ($\alpha>0$) for all $h\in V(H)$, we have $$\nu(G\times H)\geq \max \{\gamma^{2\alpha}(G),\gamma^{2\alpha}(H)\}.$$
\end{theorem}
\begin{proof}
Let $S$ be a minimal dominating set of $G\times H$. Define $f:V(G)\rightarrow \mathbb{R}$ such that $$f(g)=\min\{2\alpha, |S\cap ~^{g}H|\}.$$

We prove that $f$ is a $2\alpha$ dominating function. Let $g\in V(G)$ and $V(H)=\{h_{1},h_{2},\ldots,h_{n}\}$.

Case I: If $f(g)=2\alpha$, then clearly $f(N[g])\geq 2\alpha$.

Case II: Let $\alpha\leq f(g)< 2\alpha$ (i.e., $S\cap ~^{g}H$ contains exactly one element).

Without loss of generality we can assume $(g,h_{1})\in S$. Then $(g,h_{i})\notin S$ for $i\geq 2$. Since $S$ is a dominating set, $(g,h_{2})$ is dominated by some $(x,h_{j})\in S$ where $x\neq g$. Thus $f(x)\geq \alpha$. Also since $x$ is adjacent to $g$, we have $f(N[g])\geq 2\alpha$.

Case III: $f(g)=0$ (i.e., $S\cap ~^{g}H$ is an empty set).

In this case $(g,h_{1})\notin S$ and $(g,h_{2})\notin S$ and so, there exist $g_{1},g_{2}\in G$ and $h_{i},h_{j}\in H$ such that $(g_{1},h_{i})\in S$ dominates $(g,h_{1})$ and $(g_{2},h_{j})\in S$ dominates $(g,h_{2})$. Now if $x_{1}\neq x_{2}$, then $f(x_{1})\geq \alpha$ and $f(x_{2})\geq \alpha$ and so, $f(N[g])\geq 2\alpha$.

Again suppose $x_{1}= x_{2}$. Now if $i\neq j$, then $f(x_{1})\geq 2\alpha$ and thus $f(N[g])\geq 2\alpha$.

Finally suppose $g_{1}=g_{2}$ and $h_{i}=h_{j})$, i.e., $(g,h_{1})$ and $(g,h_{2})$ are both dominated by $(g_{1},h_{i})$. Clearly, $i\neq 1,2$. Then $(g,h_{i})$ must be dominated by some $(u,h_{p})$ where $p\neq i$. If $u=g_{1}$, then $f(g_{1})\geq 2\alpha$ and so $f(N[g])\geq 2\alpha$.

Again if $u\neq g_{1}$, then $f(g_{1})\geq \alpha$ and $f(u)\geq \alpha$ and thus $f(N[g])\geq 2\alpha$.

Therefore $f$ is a $2\alpha$-dominating function of $G$ with $w(f)\leq fc(S)$ and hence $\nu(G\times H)\geq \gamma^{2\alpha}(G)$.

In a similar fashion it can be proved that $\nu(G\times H)\geq \gamma^{2\alpha}(H)$.
\end{proof}

\proof[Conclusion] 
The theory of domination and total domination in product graphs is a rich topic of investigation both from the point of view of theory and application. Computer science is a branch in which the use of product graphs have become indispensable, specifically, in problems such as load balancing for massively parallel computer architectures. In this paper we have defined the concept of direct product of two fuzzy graphs and proved some important results related to the theory of domination in this new set up. The concepts have also been illustrated with precise examples. Other graphical invariants such as independence and graph coloring will be good subjects to investigate in the similar setting of direct product fuzzy graph.  
%


\bibliographystyle{plain}
\bibliography{dom_prod}

\end{document}